\documentclass{article}

\usepackage{amsmath}
\usepackage{amsfonts}
\usepackage{amssymb}
\usepackage{bm}
\usepackage{multirow}
\usepackage{subfigure}
\usepackage{natbib}
\usepackage{siunitx}

\RequirePackage[thmmarks,amsmath,amsthm]{ntheorem}
\theoremstyle{plain}
\newtheorem*{proposition}{Proposition}
\newtheorem*{lemma}{Lemma}

\theoremstyle{remark}
\newtheorem*{example}{Example}

\newcommand{\pd}[2]{\frac{\partial#1}{\partial#2}} % partial derivatives macro
\newcommand{\bmSigma}{\mbox{\boldmath$\sigma$}}
\newcommand{\bmx}{{\bf x}}

\newcommand{\bmbeta}{\boldsymbol{\beta}}

\newcommand{\dnorm}{\mathrm{N}}

\newcommand{\eg}{{\it e.g.}}

\begin{document}

\title{Langevin diffusions and the Metropolis-adjusted Langevin algorithm}
\author{T.~Xifara$^{1,2}$\footnote{Correspondence author: xifara@ams.ucsc.edu, Department of Applied Mathematics \& Statistics, University of California, Santa Cruz, CA 95064, USA} \and C.~Sherlock$^1$
%\ead{c.sherlock@lancaster.ac.uk
\and S.~Livingstone$^3$
%\ead{samuel.livingstone@ucl.ac.uk}
\and S.~Byrne$^3$
%\ead{simon.byrne@ucl.ac.uk}
\and M.~Girolami$^3$
%\ead{m.girolami@ucl.ac.uk}
\\
$^1$Department of Mathematics \& Statistics, Lancaster University, UK \\
$^2$Applied Mathematics \& Statistics, University of California, Santa Cruz, USA\\
$^3$Department of Statistical Science, University College London, UK 
}
\date{}
\maketitle

\begin{abstract}
We provide a clarification of the description 
of Langevin diffusions on Riemannian manifolds and of 
 the measure underlying the invariant
density. As a result we propose a new position-dependent 
Metropolis-adjusted Langevin algorithm (MALA)
 based upon a Langevin diffusion in $\mathbb{R}^d$ which has the required invariant density with respect to Lebesgue measure.
We show that our diffusion and the diffusion upon which a previously-proposed
position-dependent MALA is based are equivalent in some cases but are distinct in general.
A simulation study illustrates the gain in efficiency provided by
the new position-dependent MALA.
\end{abstract}


%% Main Body %%

\section{Introduction}
The Metropolis-adjusted Langevin algorithm (MALA)
\cite[\eg][]{Roberts1998} and its manifold variant
(MMALA) \cite[][]{Girolami2011} are Markov chain Monte Carlo methods based on diffusions.  While theoretical properties of the former are better
understood \cite[\eg][]{Roberts1998}, the latter has been shown to be more
effective in practice, producing more efficient estimates for the same
computational budget in many experiments \cite[][]{Girolami2011}.  In
this article we highlight two properties of the diffusion on which
MMALA is based.  First, we point out an unfortunate transcription
error which has propagated through the literature, whereby a factor of
a 1/2 has been missed from one of the terms \cite{Roberts2002,Girolami2011}.  Second, we show that the \emph{corrected}
diffusion does not have the intended invariant density with respect to
Lebesgue measure.  It would seem logical that a similar diffusion
which \emph{does} preserve the intended probability density may prove
a better basis for a Metropolis--Hastings algorithm.  We therefore
describe such a diffusion and the resulting sampling method, which we
call PMALA (position-dependent MALA).  We show that the incorrectly
transcribed diffusion and that on which PMALA is based are equivalent
in some cases, although the former leads to  a more computationally costly
algorithm; this equivalence explains to some extent why the error has been
missed previously.  Finally we describe simulation studies based on those in
\cite{Girolami2011} comparing PMALA and MMALA. In terms of effective
sample size (ESS) PMALA outperforms  MMALA when the two are not
equivalent. PMALA always outperforms MMALA in terms of
effective sample size (ESS) per second, 
since even when the two algorithms are equivalent, each step of 
PMALA involves fewer CPU operations. 

\section{Langevin diffusions}
A $d$-dimensional diffusion is a continuous time stochastic process
$\bm{X} = (\bm{X}_t)_{t \geq 0}$ with almost surely continuous sample
paths.  It can be (formally) written as a solution to a stochastic
differential equation $d\bm{X}_t = \bm{b}(\bm{X}_t)dt +
\bmSigma(\bm{X}_t)d\bm{W}_t$ for drift vector $\bm{b}(\bm{x})$ and volatility matrix
$\bmSigma(\bm{x})$, and where $\bm{W}=(\bm{W}_t)_{t \geq 0}$ is a standard $d$-dimensional Wiener process \cite[\eg][]{Durrett1996}.  
Given an initial condition $\bm{X}_0 = \bm{x}_0$, a realisation can be approximately
simulated using numerical techniques.  The Euler--Maruyama method
\cite[\eg][]{Kloeden1992} is among the simplest: for a chosen
step-size $h$ a realisation the sequence of random variables
$\bm{X}_h,\bm{X}_{2h},...,\bm{X}_{nh}$ is approximated using the
procedure 
%\begin{equation*}
 $\bm{x}_{(i+1)h} = \bm{x}_{ih} + h\bm{b}(\bm{x}_{ih}) + \bmSigma(\bm{x}_{ih})\bm{\varepsilon}$,
%\end{equation*}
where $\bm{\varepsilon} \sim \dnorm_d(0,hI_{d \times d})$. 

The law of the diffusion is described by the Fokker--Planck equation
\cite[\eg][]{Oksendal1998}, which relates the evolution of the
probability density function $u(\bm{x},t)$ for $\bm{X}_t$ to the drift and volatility $\bm{b},\bmSigma$,
\begin{equation} \label{eqn:fpk}
\frac{\partial}{\partial t} u(\bm{x},t) = -\sum_i \frac{\partial}{\partial x_i} [b_i(\bm{x})u(\bm{x},t)]+ \frac{1}{2}\sum_{i,j} \frac{\partial^2}{\partial x_i \partial x_j}[V_{ij}(\bm{x})u(\bm{x},t)],
\end{equation}
where $V(\bm{x}) = \bmSigma(\bm{x})\bmSigma(\bm{x})^T$. If
$u(\bm{x},t) = \pi(\bm{x})$ for all $t$, then the process is
\emph{stationary}, and $\pi$ is the density of the \emph{invariant} or
\emph{stationary distribution} of the diffusion,
%\footnote{Throughout this article $\pi$ (or $\pi(\bm{x})$) refers to
%a density, which unless otherwise specified is with respect to
%Lebesgue measure, while $\pi(\cdot)$ is a distribution with density
%$\pi$.}
 meaning that if $\bm{X}_t \sim \pi(\cdot)$ then $\bm{X}_{t+\tau} \sim \pi(\cdot)$ for all $\tau > 0$ \cite[\eg][]{Oksendal1998}.  One such is the Langevin diffusion, the solution to:
\begin{equation} \label{eqn:lang}
d\bm{X}_t = \frac{1}{2}\nabla\log\pi(\bm{X}_t)dt + d\bm{W}_t, \: \bm{X}_0 = \bm{x}_0.
\end{equation}

Setting $\bm{b}$
and $\bmSigma$ as in (\ref{eqn:lang}) and $u(\bm{x},0) = \pi(\bm{x})$
gives $\partial u/\partial t = 0$, meaning that the invariant measure for the Langevin diffusion has associated density $\pi(\bm{x})$ with respect to Lebesgue measure on $\mathbb{R}^d$. Under certain conditions, the Langevin diffusion converges to $\pi$ at an exponential rate from any starting point \cite{Roberts1995}.

The Metropolis--Hastings algorithm simulates from a Markov chain which
has a desired invariant density, $\pi(\bm{x})$.  Expectations from
this distribution can be approximated by averaging values across the
chain \cite[\eg][]{gilks}.  At each iteration some proposal $\bm{x}'$
is drawn from a distribution with density $q(\cdot|\bm{x})$ (where $\bm{x}$ represents the current value in the chain).  The next value in the chain is set to be $\bm{x}'$ with probability $\alpha(\bm{x}'|\bm{x})$, or else $\bm{x}$, where:
\begin{equation} \label{eqn:ar}
\alpha(\bm{x}'|\bm{x}) = 1 \wedge \frac{\pi(\bm{x}')q(\bm{x}|\bm{x}')}{\pi(\bm{x})q(\bm{x}'|\bm{x})}.
\end{equation}
\textit{Any} diffusion can form the basis of a Metropolis--Hastings
algorithm
% with invariant density $\pi$, 
by setting the proposal density as 
$\bm{X}' \sim \dnorm_d(\bm{x} +h b(\bm{x}), hV(\bm{x}))$.  Since the objective is to simulate a chain with invariant density $\pi$, basing a scheme on a Langevin diffusion (\ref{eqn:lang}) which itself has invariant density $\pi$ seems logical, and indeed the diffusion (\ref{eqn:lang}) is the basis of the MALA method, whereby proposals are generated according to:
\[
\bm{X}' \sim \dnorm_d \left(\bm{x} + \frac{h}{2}\nabla\log\pi(\bm{x}), hI_{d \times d} \right),
\]
for a chosen step size $h$, and then accepted with probability $\alpha(\bm{x}'|\bm{x})$.  Scaling properties of $h$ with $d$ and asymptotic optimal acceptance rates for the method are discussed in \cite{Roberts1998}.  A slight generalisation of (\ref{eqn:lang}) is the diffusion:
\begin{equation} \label{eqn:prelang}
d\bm{X}_t = \frac{1}{2}A\nabla\log\pi(\bm{X}_t)dt + \surd{A}~d\bm{W}_t,
\end{equation}
where $A$ is a positive-definite matrix, and 
%$\surd{A}$ is any matrix such that 
$\surd{A}\surd{A}^T = A$.  As with the diffusion (\ref{eqn:lang}), substitution of the drift and volatility terms from (\ref{eqn:prelang}) into the Fokker--Planck equation leads to $\partial u/\partial t=0$, so that $\pi$ is the invariant density of (\ref{eqn:prelang}).  The Metropolis--Hastings scheme derived from (\ref{eqn:prelang}) is known as `pre-conditioned MALA' \cite{Roberts2002} and is well-suited to scenarios in which the components of $\pi$ are highly correlated or have very different marginal variances, but where these relationships vary little over the main posterior mass.

The MMALA algorithm \cite{Girolami2011} is based on the discretisation of a diffusion with a position-dependent volatility matrix:
\begin{align} 
  \label{eqn:MMALAlang}
d\bm{X}_t &= \frac{1}{2}G^{-1}(\bm{X}_t)\nabla\log\pi(\bm{X}_t)dt + \Omega(\bm{X}_t)dt + \surd{G^{-1}(\bm{X}_t)}~d\bm{W}_t, \\
\Omega_i(\bm{X}_t) &= |G(\bm{X}_t)|^{-1/2} \sum_j \pd{}{X_j} \bigl[G^{-1}_{ij}(\bm{X}_t)|G(\bm{X}_t)|^{1/2} \bigr], \nonumber
\end{align}
where $G(\bm{X}_t)$ is some positive definite $d \times d$ matrix.  The choice of $G$ is arbitrary, but some natural candidates arise by noting that the above process can be thought of as a diffusion defined on a Riemannian manifold, specified in local coordinates (\cite{Girolami2011}).  In the resulting algorithm, proposals are generated according to:
\begin{equation} \label{eqn:MMALAprop}
\bm{X}' \sim \dnorm_d \left( \bm{x} + \frac{h}{2}G^{-1}(\bm{x})\nabla\log\pi(\bm{x}) + h\Omega(\bm{x}), h G^{-1}(\bm{x}) \right),
\end{equation}
and then accepted or rejected according to (\ref{eqn:ar}).  A similar
scheme is proposed in \cite{Roberts2002}, based on the same diffusion.
For a suitable choice of $G(\bm{x})$, the position-dependent covariance matrix for proposals in (\ref{eqn:MMALAprop}) allows adaptation to the local curvature of the target density $\pi$, which has been shown to increase algorithm efficiency in a number of examples \cite{Girolami2011}.

\section{A new position-dependent diffusion and MALA}
\label{sec: PMALA}
In general, a diffusion with invariant density $\pi$ can be constructed by starting from (\ref{eqn:fpk}) and selecting a drift and volatility such that 
\begin{equation} \label{eqn:FPKstat}
b_i(\bm{x})\pi(\bm{x}) = \frac{1}{2} \sum_j \pd{}{x_j} \bigl[ V_{ij}(\bm{x})\pi(\bm{x}) \bigr].
\end{equation}
 If the intention is to derive a Metropolis--Hastings proposal mechanism with a position-dependent covariance matrix, a natural starting point would be to simply set $A = A(\bm{X}_t)$ in (\ref{eqn:prelang}), giving
$d\bm{X}_t = \frac{1}{2}A(\bm{X}_t)\nabla\log\pi(\bm{X}_t)dt + \surd{A(\bm{X}_t)}~d\bm{W}_t$,
%\begin{equation} \label{eqn:prelang2}
%d\bm{X}_t = \frac{1}{2}A(\bm{X}_t)\nabla\log\pi(\bm{X}_t)dt + \surd{A(\bm{X}_t)}~d\bm{W}_t,
%\end{equation}
and this diffusion forms the basis of the simplified MMALA algorithm of \cite{Girolami2011}.  However, substituting the drift and volatility terms into (\ref{eqn:FPKstat}) gives the requirement that:
\begin{equation} \label{eqn:prelang3}
\frac{1}{2}\sum_j A_{ij}(\bm{x})\pd{}{x_j} \bigl[\log\pi(\bm{x}) \bigr]\pi(\bm{x}) = \frac{1}{2}\sum_j \left( \pd{A_{ij}(\bm{x})}{x_j}\pi(\bm{x}) + A_{ij}(\bm{x})\pd{\pi(\bm{x})}{x_j} \right)
\end{equation}
for each $i$.  Since $\partial/\partial
x_j[\log\pi(\bm{x})]\pi(\bm{x}) = \partial \pi(\bm{x})/\partial x_j$,
(\ref{eqn:prelang3}) is only satisfied in general when $A$ is a
constant matrix.  A simple modification to the drift term, however,
leads to a new diffusion which satisfies (\ref{eqn:FPKstat}):
\begin{align} 
  \label{eqn:PMALAlang}
d\bm{X}_t &= \frac{1}{2}A(\bm{X}_t)\nabla\log\pi(\bm{X}_t)dt + \Gamma(\bm{X}_t)dt + \surd{A(\bm{X}_t)}~d\bm{W}_t \\
\Gamma_i(\bm{X}_t) &= \frac{1}{2} \sum_j \pd{}{X_j} A_{ij}(\bm{X}). \nonumber
\end{align}
This diffusion has invariant density $\pi$ with respect
to Lebesgue measure, and the additional drift term $\Gamma$ is
of a simpler form than $\Omega$ in \eqref{eqn:MMALAlang}.  The
resulting Metropolis--Hastings proposal mechanism is: 
%$\bm{x}' \sim \dnorm_d \left( \bm{x} + \frac{h}{2}A(\bm{x})\nabla\log\pi(\bm{x}) + h\Gamma(\bm{x}), hA(\bm{x}) \right)$.
\[
\bm{X}' \sim \dnorm_d \left( \bm{x} + \frac{h}{2}A(\bm{x})\nabla\log\pi(\bm{x}) + h\Gamma(\bm{x}), hA(\bm{x}) \right).
\]
We refer to the resulting Metropolis--Hastings method as `position-dependent MALA' or, more succinctly, `PMALA'.

The remainder of this section details two connections between
the diffusions (\ref{eqn:MMALAlang}) and (\ref{eqn:PMALAlang}) when $A(\bm{X}_t)=G^{-1}(\bm{X}_t)$. In
describing these connections the following 
equivalent forms for the $i$th components of $\Omega(\bm{X}_t)$ and
$\Gamma(\bm{X}_t)$ will be helpful. For clarity of exposition 
%in the following 
we suppress
explicit dependence on $X_t$ of all four of these quantities.
\begin{align}
\label{eqn.Omega.alt.one}
\Omega_i&=\sum_j \frac{\partial G^{-1}_{ij}}{\partial X_j}
+\frac{1}{2}\sum_j G^{-1}_{ij}\frac{\partial \log |G|}{\partial X_j}\\
%\nonumber
%&=&
%-\sum_j \left[G^{-1}\frac{\partial G}{\partial X_j}G^{-1}\right]_{ij}+
%\frac{1}{2}\sum_j G^{-1}_{ij}\mbox{tr}\left(G^{-1}\frac{\partial G}{\partial X_j}\right)%\\
&=
\label{eqn.Omega.alt.two}
-\sum_{jkm}G^{-1}_{ik}\frac{\partial G_{km}}{\partial X_j}G^{-1}_{mj}
+\frac{1}{2}\sum_{jkm}G^{-1}_{ij}\frac{\partial G_{mk}}{\partial
  X_j}G^{-1}_{km}.\\
\nonumber\label{Gamma.alt.onetwo}
\Gamma_i&=
\frac{1}{2}\sum_j \frac{\partial G^{-1}_{ij}}{\partial X_j}
=
-\frac{1}{2}\sum_{jkm}G^{-1}_{ik}\frac{\partial G_{km}}{\partial X_j}G^{-1}_{mj}
\end{align}
The first connection arises because the diffusion (\ref{eqn:MMALAlang}) on
which both MMALA and the algorithm of \cite{Roberts2002} are based contains a
transcription error.  The term $\Omega$ should be multiplied by a
factor of $1/2$, giving the diffusion
\begin{equation} \label{eqn:MMALAlang2}
d\bm{X}_t = \frac{1}{2}G^{-1}(\bm{X}_t)\nabla\log\pi^*(\bm{X}_t)dt + \frac{1}{2}\Omega(\bm{X}_t)dt + \surd{G^{-1}(\bm{X}_t)}~d\bm{W}_t.
\end{equation}
This can be viewed as a deterministic mapping of (\ref{eqn:lang}) onto a
Riemannian manifold with metric tensor $G$, with the first term being the
covariant drift, and the second and third corresponding to a Brownian motion
on the manifold \cite{Kent1978}. However, the density $\pi^*$ is not given
with respect to the Lebesgue measure, but instead respect to the
dimensional volume or Hausdorff measure of the manifold, which is
coordinate invariant. We 
refrain from
discussing this in detail, but note that this is related to the density $\pi$
with respect to the Lebesgue measure via the \emph{area formula} \cite[Theorem
3.2.5]{Federer},
\begin{equation} \label{eqn:areaformula}
  \pi(\bm{x}) = \pi^*(\bm{x}) |G(\bm{x})|^{1/2}.
\end{equation}

\begin{lemma}
  The diffusions defined by (\ref{eqn:MMALAlang2}) and (\ref{eqn:PMALAlang}) are equal.
\end{lemma}
\begin{proof}
  The volatilities of the two diffusions are the same, so we need only compare
  the drift terms.  Substituting \eqref{eqn:areaformula} into
  (\ref{eqn:PMALAlang}) gives a diffusion where the $i$th component of the
  drift term is
  \[
  b_i=\frac{1}{2}\sum_jG^{-1}_{ij}\frac{\partial \log \pi^*}{\partial X_j}
  +\frac{1}{4}|G|\sum_jG^{-1}_{ij}\frac{\partial |G|}{\partial X_j} +
  \frac{1}{2}\sum_j\frac{\partial G^{-1}_{ij}}{\partial X_j}
  \]
  which, using (\ref{eqn.Omega.alt.one}), is the $i$th component of the drift
  in (\ref{eqn:MMALAlang2}).
\end{proof}

Thus the diffusion (\ref{eqn:MMALAlang}) arises as a result of both an
error in transcription and omitting the determinant factor when
changing reference measures. Interestingly, in certain circumstances these two
mistakes appear to cancel, so that (\ref{eqn:MMALAlang}) does, in
fact, have the correct invariant distribution.

\begin{proposition} If $G(\bm{x})$ is chosen such that for any combination of $1 \leq j,k,m \leq d$:
\begin{equation}
\label{eqn.invariance}
\pd{}{x_j}G_{km}(\bm{x}) = \pd{}{x_k}G_{jm}(\bm{x})
\end{equation}
for all $\bm{x}$, then (\ref{eqn:MMALAlang}) and (\ref{eqn:PMALAlang})
represent the same diffusion. 
\end{proposition}

\begin{proof}
  Since the volatitilites and the multipliers of $\nabla \log \pi$ in the
  drift are identical for the two diffusions, we need only show that
  $\Omega_i=\Gamma_i$ for all $i$. From (\ref{eqn.invariance}) the second term
  in (\ref{eqn.Omega.alt.two}) can be rewritten as
  \[
  \frac{1}{2}\sum_{jkm}G^{-1}_{ij}\frac{\partial G_{jm}}{\partial
    X_k}G^{-1}_{km} = \frac{1}{2}\sum_{jkm}G^{-1}_{ik}\frac{\partial
    G_{km}}{\partial X_j}G^{-1}_{jm},
  \]
  on relabelling $j\leftrightarrow k$. The result follows since
  $G^{-1}_{jm}=G^{-1}_{mj}$.
\end{proof}

This property arises in certain simple cases, which suggests, perhaps,
how this mistake has thus far remained undetected. If the process is
univariate ($d=1$), then \eqref{eqn.invariance} holds trivially. More
generally, it also holds if $G$ is the (continuous) Hessian matrix of some real-valued function: in particular, in the case of a natural exponential family, such as a generalised linear model (GLM) with canonical link, the Fisher information matrix used by \cite{Girolami2011}, is equal to the Hessian of the negative log-likelihood function.

In general, however, the diffusion (\ref{eqn:MMALAlang}) will not have the
desired invariant density.
\begin{example}
  For some positive-valued, differentiable function $f$, set
  \[
  G(\bm{x}) = \begin{bmatrix} f(x_2) & 0 \\ 0 & 1 \end{bmatrix}.
  \]
  It is then straightforward to show that $\Gamma = [0,0]$ and $\Omega = [0,
  f'(x_2) / f(x_2)]$, and hence the diffusions \eqref{eqn:MMALAlang} and
  \eqref{eqn:PMALAlang} have different drift coefficients. Moreover, the
  diffusion \eqref{eqn:MMALAlang} can be written in the same form as that of
  \eqref{eqn:PMALAlang}; by matching the drift terms, it can be seen that
  the invariant density of \eqref{eqn:MMALAlang} is actually proportional to
  $\pi(\bm{x}) f(x_2)$.
\end{example}

\section{Experiments}
We compared the performance of the MALA schemes across three of the
scenarios considered in \cite{Girolami2011}: logistic regression on
each of five different datasets; a stochastic volatility model; and a
non-linear ODE model. As in \cite{Girolami2011} we base the metric
tensor, $G(\bmx)=A(\bmx)^{-1}$, on the expected Fisher information.

Initial tuning runs provided the optimal scaling parameter(s)
($\sqrt{h}$ in this article) in terms of ESS for each algorithm (on each
dataset, where relevant).  The initialisation, burn-in, and length of each
Markov chain was exactly as in \cite{Girolami2011}, however we performed $100$
(rather than $10$) replicated runs for each chain.

Bayesian logistic regression and the non-linear ODE model are of most interest
since in \cite{Girolami2011} MMALA was found to outperform Riemann Manifold
Hamiltonian Monte Carlo for these scenarios. Due to space considerations we
therefore present detailed results for these scenarios; results for the
stochastic volatility model showed the same pattern as for the
non-linear ODE model. Where
especially pertinent we provide brief details on the models themselves and the
priors; for further details the reader is referred to \cite{Girolami2011}.

\subsection{Logistic regression}
\label{sec: bayeslog}
We perform Bayesian inference for a logistic regression model on each
of five different datasets containing between $7$ and $25$ covariates.
%: Australian Credit, German Credit,
% Heart; Pima Indian and Ripley. 
We choose a Gaussian prior for the
 parameter vector $\bmbeta \sim \mathcal{N}(\mathbf{0}, \alpha I)$, so that with a design matrix is $X$ and link function $s(\cdot)$
the metric tensor is given by
$G (\bmbeta) = X^T  \Lambda X + \alpha^{-1} I$,  where $\Lambda$ is a diagonal matrix with
elements $\Lambda_{i,i} = s(\bmbeta^T X^T_{i,\cdot})(1-
s(\bmbeta^T X^T_{i,\cdot}))$. As noted above, this satisfies (\ref{eqn.invariance}), so the diffusions on which PMALA and MMALA are based have
the same law and we should expect the ESSs for these two algorithms to
be the same up to Monte Carlo error.

For each Markov chain the ESS was computed for each parameter and the
minimum, median and maximum of these was noted. Table
\ref{tbl:bayes_log} shows, for each algorithm and dataset, the means
and their corresponding standard errors using 
the $100$ replicates. The CPU time and the mean (over replicates)
minimum (over parameters) effective number of independent
samples per second are also provided.

As expected, the ESSs for PMALA and MMALA are very similar. Since $\Gamma$ is
computationally less costly to calculate than $\Omega$, PMALA is quicker and so obtains the larger ESS per second.

\begin{table}[h]
\centering
\caption{\label{tbl:bayes_log} Results for the MMALA schemes for Bayesian
  logistic regression. The mean (over the $100$
  replicates) and its 
  standard error is presented for the minimum,
  median and maximum ESSs (over the parameters). The CPU time and the
  mean minimum ESS per second are also given.}
\fbox{ 
\begin{tabular}{l l c c S S}
\multirow{2}{*}{\it Dataset} & \multirow{2}{*}{\it Method} & \multirow{2}{*}{\it ESS (mean)} & \multirow{2}{*}{\it ESS (s.e)} &\multirow{2}{*}{\it
  CPU Time} & {\it minimum} \\
  & & & && {\it  ESS/s}  \vspace{0.02cm}\\ 
\hline 
Australian & PMALA & (685, 847, 986) & (5.5, 3, 4.1) & 12.58 & 54.5 \\ 
Credit & MMALA & (696, 848, 943) & (6, 2.9, 4.1) & 14.08 & 49.4 \\
\hline 
German & PMALA & (605, 777, 917) & (5.4, 2.5, 4) & 43.8 & 13.8 \\
Credit & MMALA & (605, 774, 921) & (5.5, 2.5, 3.9) & 45.72 & 13.2 \\
\hline
\multirow{2}{*}{Heart} 
& PMALA &  (659, 795, 923) & (5.4, 3.3, 4.3) & 6.57 & 100.3 \\
& MMALA & (657, 773, 920) & (4.8, 2.9, 4.7) & 8.07 & 81.4 \\
\hline
Pima & PMALA & (1235, 1415, 1572) & (8.7, 5.9, 6.6) & 4.67 & 264.5 \\
Indian & MMALA & (1264, 1425, 1576) & (9.6, 6.5, 7.6) & 5.59 & 226.1 \\
\hline
\multirow{2}{*}{Ripley} 
& PMALA & (477, 591, 679) & (6.8, 5.1, 5) & 3.32 & 143.7 \\
& MMALA & (460, 590, 686) & (7.5, 5.2, 5.3) & 3.94 & 116.7 
\end{tabular}
}
\end{table}

\subsection{Non-linear differential equation model}
We now consider the FitzHugh--Nagumo differential equations in 
\cite{Ramsay2007}:
$\dot{W}  = c \left(W - {W^3}/{3}+R \right)$ and 
$\dot{R}  = - (W-a+bT)/{c}.$ 
%\[
%\dot{W}  = c \left(W - \frac{W^3}{3}+R \right),~~~
%\dot{R}  = - \left( \frac{W-a+bT}{c}\right). 
%\]
The simulated dataset and our independent priors for the parameter
vector $(a,b,c)$ and the variance of the Gaussian observation noise are the same as those used in
\cite{Girolami2011}. To be consistent with the appendix of
\cite{Girolami2011} and the associated  Matlab code we assume
$\beta\sim \mbox{Exp}(1)$.

 \begin{table}
 \centering
\caption{\label{tbl: ODE} Results for the MMALA schemes for inference on
  the FitzHugh--Nagumo model. For each parameter (a,b,c) and algorithm the
  mean (over the $100$ replicates) ESS is presented as well as its
  standard error. The CPU time and the mean ESS per second for each
  parameter are also provided.}
\fbox{
\begin{tabular}{l c c c c}
{\it Method} & {\it ESS (mean)} & {\it ESS (s.e)} & {\it CPU Time} &
{\it mean ESS/s}  \vspace{0.05cm}\\ 
\hline 
PMALA & (1639.6, 669.3, 1406.4) & (1.9, 1.2, 1.7) & 896.8 & (1.83, 0.75, 1.57) \\
MMALA & (1274.4, 632.8, 1120.5) & (1.7, 1.2, 1.3) & 923.0 & (1.38, 0.69, 1.21) 
\end{tabular}
}
\end{table}
Table \ref{tbl: ODE} presents the mean ESS for each parameter with its
standard error and shows that PMALA outperforms MMALA using this
measure. CPU time and ESS/sec are also provided in the table; since each
iteration of PMALA is also quicker, its advantage is even clearer when CPU
time is accounted for.

\section*{Acknowledgements}
T. Xifara was part-funded by North West Development Agency project N0003235  and the Greek State Scholarships Foundation.
S. Livingstone is funded by a PhD Scholarship from Xerox Research
Centre Europe.  S. Byrne is funded by an EPSRC Postdoctoral Research
Fellowship, EP/K005723/1.  M. Girolami is funded by an EPSRC
Established Career Research Fellowship, EP/J016934/1 and a Royal
Society Wolfson Research Merit Award and is grateful to Prof. Jesus Sanz Serna for drawing his attention to the underlying Hausdorff measure of the diffusion (5) in a personal communication.

%\section*{References}
\bibliographystyle{chicago} 
\bibliography{references_MALA}

\end{document}